\documentclass[10pt,conference]{IEEEtran}

\usepackage{amsfonts}
\usepackage{amssymb}
\usepackage{cite}
\usepackage[cmex10]{amsmath}
\usepackage{float}
\usepackage{color}
\usepackage{stfloats,fancyhdr}

\usepackage{algorithm}
\usepackage{algorithmic}
\usepackage{multirow}
\usepackage{changepage}
\usepackage[normalem]{ulem}

\usepackage{geometry}
\newtheorem{lemma}{Lemma}
\newtheorem{proposition}{Proposition}

\newtheorem{remark}{Remark}

\IEEEoverridecommandlockouts

\ifCLASSINFOpdf
  \usepackage[pdftex]{graphicx}
  \DeclareGraphicsExtensions{.pdf,.jpeg,.png}
\else
  \usepackage[dvips]{graphicx}
  \DeclareGraphicsExtensions{.eps}
\fi

\usepackage{subfigure}

\usepackage{fancybox,dashbox}

\begin{document}

\title{
Wireless-Powered Cooperative Communications via a Hybrid Relay
\thanks{The work was supported by the Australian Research Council (ARC) under Grants DP120100190 and
FT120100487, International Postgraduate Research Scholarship (IPRS), Australian Postgraduate Award (APA), and Norman I Price Supplementary scholarship. The work of X. Zhou was supported under Australian Research Council's Discovery Projects funding scheme (project number DP140101133).}
}


\author{\IEEEauthorblockN {He (Henry) Chen\IEEEauthorrefmark{1}, Xiangyun Zhou\IEEEauthorrefmark{2}, Yonghui Li\IEEEauthorrefmark{1}, Peng Wang\IEEEauthorrefmark{1}, Branka Vucetic\IEEEauthorrefmark{1}
}
\IEEEauthorblockA{\IEEEauthorrefmark{1}The University of Sydney, Sydney, Australia, Email: \{he.chen,~yonghui.li,~peng.wang,~branka.vucetic\}@sydney.edu.au}
\IEEEauthorblockA{\IEEEauthorrefmark{2}The Australian National University, Canberra, Australia, Email: xiangyun.zhou@anu.edu.au}
}


\maketitle

\begin{abstract}
In this paper, we consider a wireless-powered cooperative communication network, which consists of a hybrid access-point (AP), a hybrid relay, and an information source. In contrast to the conventional cooperative networks, the source in the considered network is assumed to have no embedded energy supply. Thus, it first needs to harvest energy from the signals broadcast by the AP and/or relay, which have constant power supply, in the downlink (DL) before transmitting the information to the AP in the uplink (UL). The hybrid relay can not only help to forward information in the UL but also charge the source with wireless energy transfer in the DL. Considering different possible operations of the hybrid relay, we propose two cooperative protocols for the considered network. We jointly optimize the time and power allocation for DL energy transfer and UL information transmission to maximize the system throughput of the proposed protocols. Numerical results are presented to compare the performance of the proposed protocols and illustrate the impacts of system parameters.
\end{abstract}

\begin{IEEEkeywords}
Wireless energy transfer, RF energy harvesting, cooperative communications.
\end{IEEEkeywords}

\IEEEpeerreviewmaketitle

\section{Introduction}

Cooperative communication technique has attracted enormous interests over the past few years due to its various advantages, such as increasing system capacity, coverage and energy efficiency \cite{Laneman_TIT_2004,Yonghui_book_2010}. By allowing nodes to share their antennas and transmit cooperatively as a virtual multiple-input multiple-output (MIMO) system, the spatial diversity can be achieved without the need of multiple antennas at each node. In practice, supportive relay stations have been deployed to improve the performance of cellular networks, WLANs and wireless sensor networks \cite{Yonghui_book_2010}.

On the other hand, radio frequency (RF) energy harvesting technique has recently emerged as a viable solution to prolong the lifetime of energy constrained wireless networks due to some significant advances of wireless power technologies \cite{Shinohara_mag_2011_Power}.
As such, a new type of wireless networks, termed wireless-powered communication network (WPCN), has attracted more and more attentions recently. In WPCNs, wireless terminals are powered only via WET and transmit their information using the harvested energy \cite{Ju_TWC_2014}. As shown in Fig. \ref{fig:system_model}, in this work we consider a simple WPCN consists of one hybrid AP, one hybrid relay, and one source node that wants to transmit its information to the AP. The AP and relay are connected to constant power supply, while the source is assumed to have no embedded energy source. But it is equipped with a rechargeable battery and thus can harvest and store the wireless energy broadcast by the AP and/or relay. In such a network, the relay plays two equally important roles. Besides the uplink (UL) information forwarding as the conventional relay, the hybrid relay also helps the AP to charge the source via WET in the downlink (DL). This is in contract to the existing papers that considered WET in cooperative networks (e.g., \cite{Krikidis_CL_2012_RF,Zhou_TWC_2013,Chen_arXiv_2014_harvest,Chen_arXiv_2014_Distributed}), where the relay was assumed to have no embedded power supply and need to harvest energy from other nodes.

\begin{figure}
\centering \scalebox{0.6}{\includegraphics{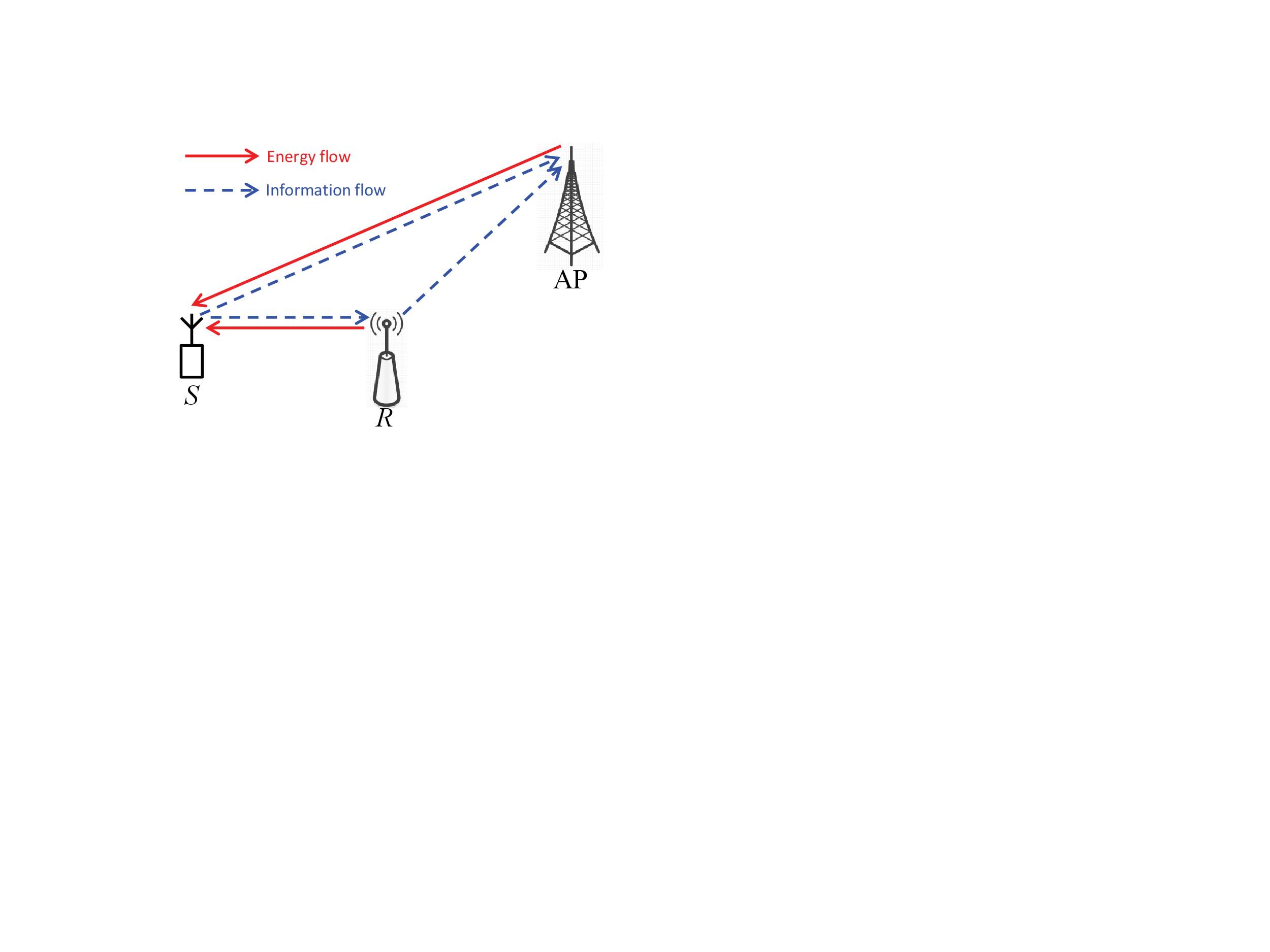}}
\caption{System model for wireless-powered cooperative communications via a hybrid relay. \label{fig:system_model}}
\end{figure}

A natural question that arises in the considered network is ``\emph{What is the optimal way to utilize the constant-powered relay for energy charging and/or information forwarding?}" This is actually an open and non-trivial question to answer. The reason is that the designs of the relay's power allocation for energy charging and/or information forwarding, the time allocation for DL energy transfer and UL information transmission, and the AP transmit power, are tangled together.


To tackle this open problem, in this paper we develop two cooperative protocols with different relay operations for the considered WPCN. Furthermore, we formulate optimization problems to maximize the system throughput by jointly designing the time allocation and power allocation for the two proposed protocols, respectively. The optimal solutions are subsequently derived and compared by simulations.  Numerical results show that the two proposed protocols can outperform each other in different network scenarios, which provides useful insights into the design of the hybrid relay in WPCNs.

\section{System Model and Protocols Description}
As shown in Fig. \ref{fig:system_model}, this paper considers a wireless-powered cooperative communication network. It is assumed that all the nodes are equipped with single antenna and work in the half-duplex mode. The source is assumed to have no embedded energy supply and thus needs to first harvest energy from the signal broadcasted by the hybrid AP and/or the relay in the DL, which can be stored in a rechargeable battery and then used for the UL information transmission.

In the sequel, we use subscript $A$ for AP, $S$ for source, and $R$ for relay. We use ${f_{XY}}
$ to denote the channel coefficient from $X$ to $Y$ with $X,Y \in \left\{ {A,S,R} \right\}$. 
The channel power gain is thus given by ${h_{XY}} = {\left| {{{f}_{XY}}} \right|^2}$. 
In addition, it is assumed that all channels in both DL and UL experience independent slow and frequency flat fading, where the channel gains remain constant during each transmission block (denoted by $T$) but change independently from one block to another.
\begin{figure}
\centering
 \subfigure[E-C Protocol]
  {\scalebox{0.22}{\includegraphics {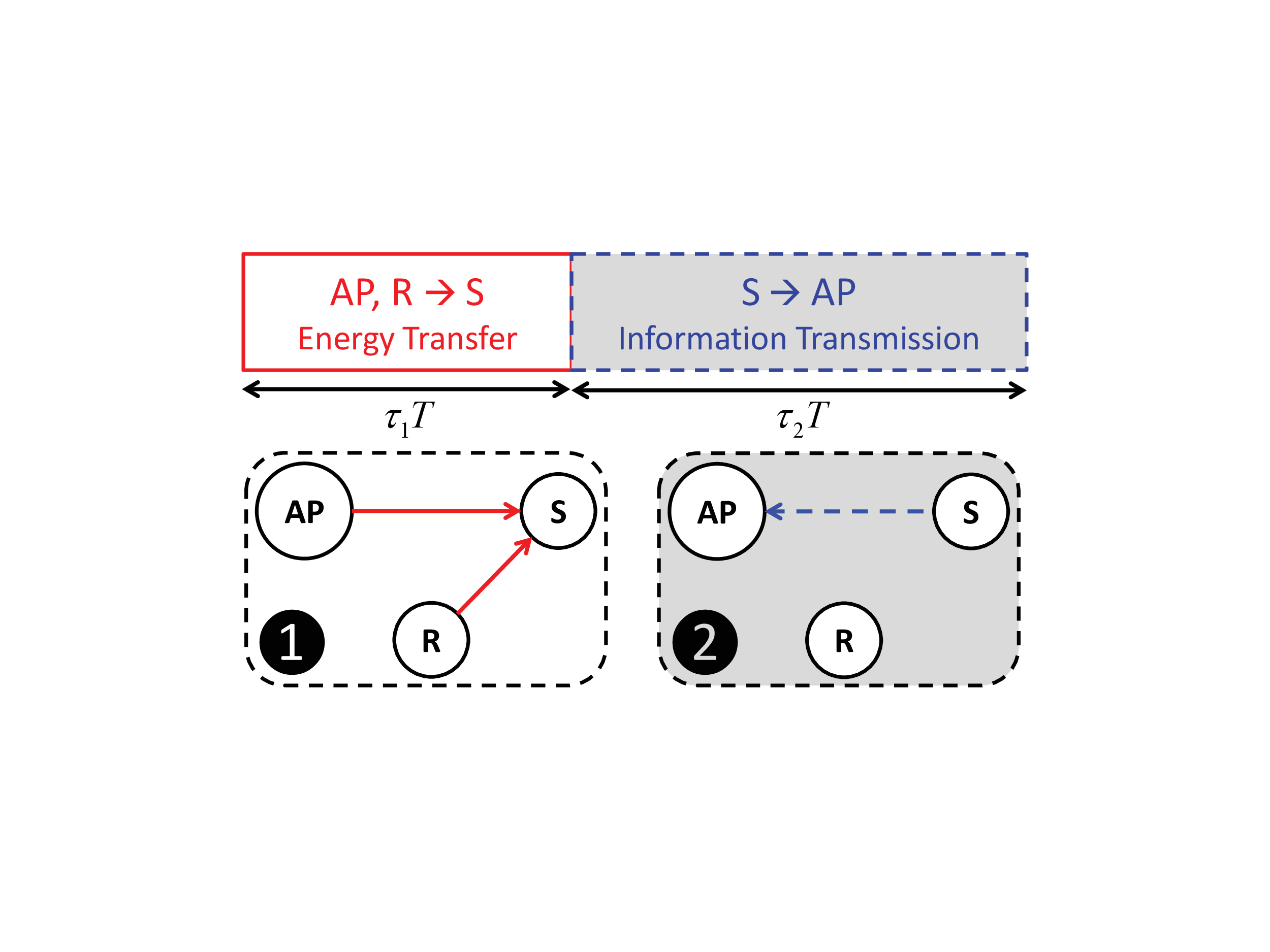}
  }}
\hfil
 \subfigure[D-C Protocol]
  {\scalebox{0.22}{\includegraphics {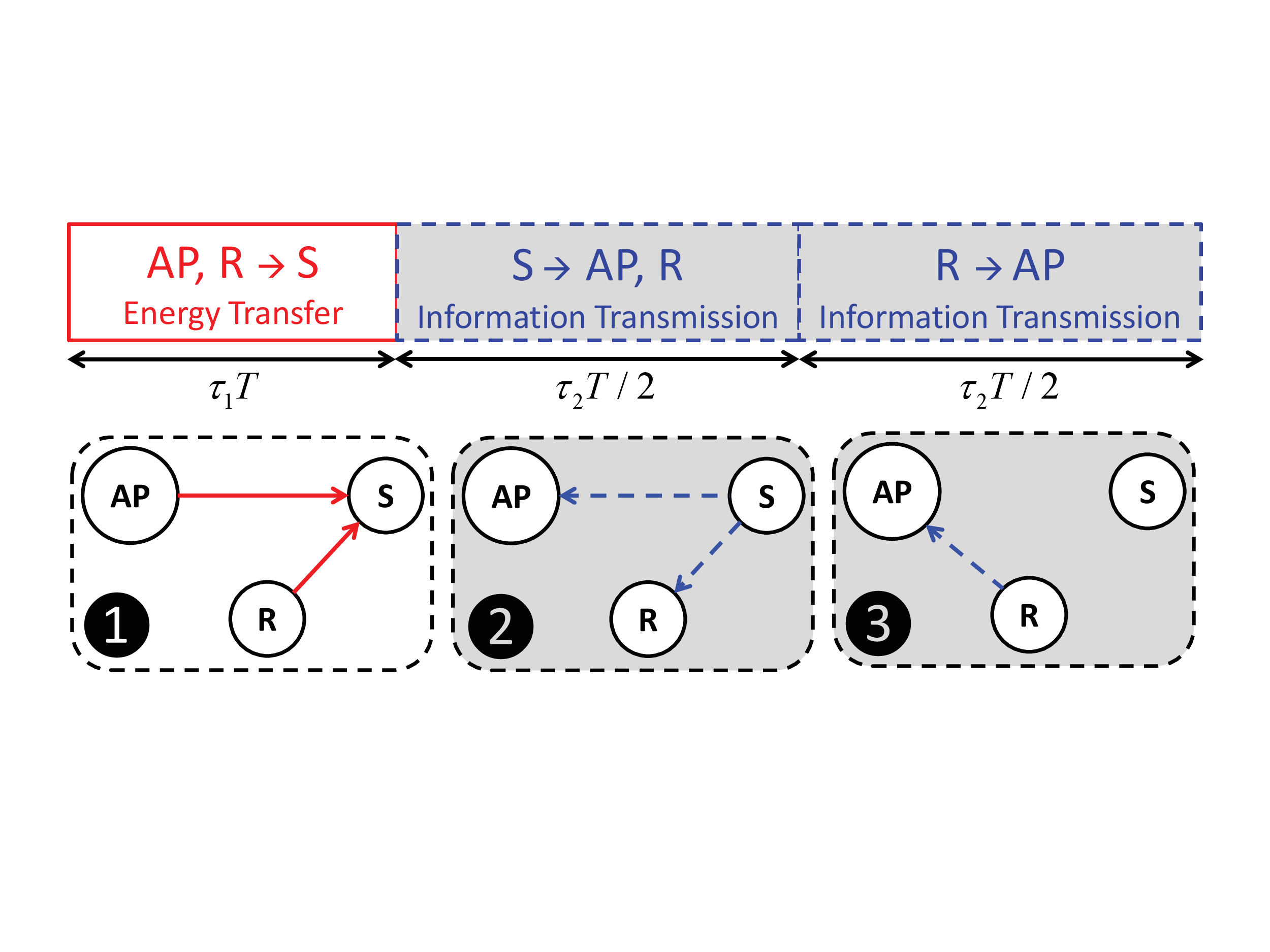}
  }}
\caption{The block diagrams for the two proposed cooperative protocols.}
\label{fig:protocols}
\end{figure}

In this paper, we develop two cooperative protocols for the considered network, referred to as \emph{energy cooperation (E-C)}\footnote{It is worth mentioning that the term ``energy cooperation" was first used in \cite{Gurakan_Tcom_2013_Energy}, where energy cooperation is used to term the following protocol: all nodes harvest some amount of energy from nature, and source node sends some energy to the relay, which in return forwards source's data via user cooperation to the destination.} and \emph{dual cooperation (D-C)}, which are different in relay operations during each transmission block, as shown in Fig. \ref{fig:protocols}. In the E-C protocol, the relay simply cooperates with the AP for DL energy transfer. In the D-C protocol, the relay first cooperates with the AP for energy transfer in the DL and then cooperates with the source for information transmission in the UL. Thus, we name this protocol as D-C (i.e., both energy and information cooperation) protocol.
In the subsequent subsections, we describe the proposed protocols and analyze their end-to-end signal-to-noise ratios (SNRs).
\subsection{E-C Protocol}
In the E-C protocol shown in Fig. \ref{fig:protocols} (a), the first $\tau_1 T$ amount of time with $0\le\tau_1 \le 1$ is assigned to the DL energy transfer, during which the AP and the relay transmit concurrently to charge the source with WET. In the following $\tau_2$ fraction of the block, the source will use the harvested energy to send its information to the AP, while the relay remains idle.

Let $P_A$ and $P_R$ denote the transmit power of the AP and relay, respectively. Here, we consider that the AP and relay have both peak and average power constraints. Mathematically, we have
\begin{equation}\label{eq:power_constraint_AP}
{P_A} \le P_A^{\max },~{\tau _1}{P_A} \le {P_A^{\rm{avg}}},
\end{equation}
\begin{equation}\label{eq:power_constraint_relay}
{P_R} \le P_R^{\max },~{\tau _1}{P_R} \le {P_R^{\rm{avg}}},
\end{equation}
where $P_X^{\max }$ and $P_X^{\rm{avg}}$ are the peak power and average power of the node $X$, $X\in \{A,R\}$. For simplicity, we consider
\begin{equation}\label{eq:def_mu}
{{{P_A^{\rm{avg}}}}}/{{P_A^{\max }}} = {{P_R^{\rm{avg}}}}/{{P_R^{\max }}} = \mu.
\end{equation}
In general, the average power should be no larger than the peak power. Hence, we have $\mu \le 1$. Note that the analytical method proposed in this paper can be readily extended to the case ${{{P_A^{\rm{avg}}}}}/{{P_A^{\max }}} \neq {{P_R^{\rm{avg}}}}/{{P_R^{\max }}}$.

Besides, $x_A^E$ and $x_R^E$ are used to denote the randomly generated energy signals with unit average energy (i.e., $\mathbb{E}\left\{ {{{\left| {{x_A^E}} \right|}^2}} \right\}= \mathbb{E}\left\{ {{{\left| {{x_R^E}} \right|}^2}} \right\} = 1$) transmitted by the AP and the relay. Then, the received signal at the source during the DL phase, denoted by $y_S$, can be expressed as
\begin{equation}\label{}
{y_S} = \sqrt {{P_A}} {f_{AS}}{x_A^E} + \sqrt {{P_R}} {f_{RS}}{x_R^E} + {n_S}.
\end{equation}
where ${n_S}$ is the additive white Gaussian noise (AWGN) at the source. The energy harvesting receiver at the source rectifies the received RF signal directly and obtains the direct current to charge up its batteries \cite{Zhou_Tcom_2013}. Moreover, we consider that the noise power is too small and below the sensitivity of energy harvesting device. Thus, the amount of energy harvested by the source in the E-C protocol is given by
\begin{equation}\label{eq:harvested_energy_source}
{E_S} = \eta \tau_1 T \left({P_A}{h_{AS}} + {P_R}{h_{RS}}\right),
\end{equation}
where $0<\eta<1$ is the energy harvesting efficiency. It is worth emphasizing that phase synchronization between the AP and relay is not required for the WET in the DL since they transmit independent energy signals. For convenience but without loss of generality, we consider a normalized unit block time (i.e., $T=1$) hereafter.

After the source replenishes its energy during the DL phase, it transmits its information to the AP by itself in the subsequent UL phase. It is assumed that the source exhausts the harvested energy for the information transmission. The transmission power of the source during the UL phase in this protocol is thus given by
\begin{equation}\label{eq:source_UL_power}
{P_S^{E-C}} = {{{E_S^{E-C}}}}/{{ {\tau_2} }}.
\end{equation}

Therefore, the end-to-end SNR at the hybrid AP in the E-C protocol can be expressed as
\begin{equation}\label{eq:SNR_E-C}
{\gamma _{E-C}} = \frac{{{P_S^{E-C}}{{ {{h_{SA}}} }}}}{{{N_0}}} = \frac{{\eta \tau_1 \left( {P_A{h_{AS}} + P_R{h_{RS}}} \right){h_{SA}}}}{{\tau_2{N_0}}},
\end{equation}
where $N_0$ denotes the power of the noise suffered by all receivers in the considered network.

\subsection{D-C Protocol}
The D-C protocol is shown in Fig. \ref{fig:protocols} (b). Analogous to the E-C protocol, the first $\tau_1 T$ amount of each transmission block is allocated for the DL energy transfer from the AP and relay to the source. The subsequent $\tau_2$ fraction of the block is further divided into two time slots with equal length of $\tau_2 T/2$ for cooperative information transmission in the UL. During the first time slot of the UL phase, the source uses the harvested energy to transmit data information to the AP, which can also be overheard by the relay due to the broadcasting feature of wireless communication. In the second time slot of the UL phase, the relay will help to forward the source's information using the amplify-and-forward (AF) relaying protocol due to its lower complexity\footnote{For the purpose of exposition, the possibility of the source harvesting energy during the relay's transmission is not taken into account in this paper. This is regarded as our future work.} \cite{Laneman_TIT_2004}.  At the end of each block, the AP combines the signals received in the first and second time slots using maximum ratio combining (MRC) technique and performs coherent detection.

Let $P_R^D$ and $P_R^U$ denote the transmit power of the relay during the DL and UL phases, respectively. Then, the peak and average power constraints for the relay in (\ref{eq:power_constraint_relay}) can be re-written as
\begin{subequations}\label{eq:power_constranit_Relay_D-C}
\begin{align}
&P_R^D \le P_R^{\max },~P_R^U \le P_R^{\max } ,\\
&{\tau _1}P_R^D + {{{\tau _2}}}P_R^U/{2} \le {P_R^{\rm{avg}}}.
\end{align}
\end{subequations}

Following the similar analysis for E-C protocol, we can readily obtain that the received SNR at the AP from the source in this protocol can be expressed as
\begin{equation}\label{}
\gamma _{SA} =  {{2\eta \tau_1 \left( {P_A{h_{AS}} + P_R^D{h_{RS}}} \right){h_{SA}}}}/\left({{\tau_2{N_0}}}\right).
\end{equation}

The received SNR at the hybrid AP from the link $S$-$R$-$A$ can thus be written as \cite{Ikki_CL_2007}
\begin{equation}\label{eq:exact_gamma_SRA}
{\gamma _{SRA}} = \frac{{{\gamma _{SR}}{\gamma _{RA}}}}{{{\gamma _{SR}} + {\gamma _{RA}} + 1}},
\end{equation}
where 
\begin{equation}\label{}
\gamma _{SR} = {{2\eta \tau_1 \left( {P_A{h_{AS}} + P_R^D{h_{RS}}} \right){h_{SR}}}}/\left({{\tau_2{N_0}}}\right),
\end{equation}
\begin{equation}\label{}
 \gamma_{RA} = {{P_R^U{h_{RA}}}}/{{{N_0}}}.
\end{equation}

Since the MRC technique is adopted at the AP receiver, the end-to-end SNR of the D-C protocol is given~by
\begin{equation}\label{eq:SNR_e2e_D-C}
{\gamma _{D-C}} = {\gamma_{SA}} + {\gamma_{SRA}}.
\end{equation}

It is worth mentioning that there exists another possible scheduling of the hybrid relay. That is, the relay keeps silent during the DL phase and only cooperates with the source for UL information transmission. However, this protocol can be regarded as a special of the D-C protocol by setting $P_R^D = 0$, which is thus omitted.

\section{Throughput Maximization for the Proposed Protocols}
In this section, we design the joint time and power allocation for the two proposed protocols to maximize their corresponding throughput. For the purpose of exposition, full channel state information (CSI) is assumed to be known.

\subsection{Throughput Maximization for E-C Protocol}
The throughput (bps/Hz) of E-C protocol can be expressed as 
\begin{equation}\label{eq:throughput_E-C}
{{\mathcal T}_{E - C}} =  \tau_2 {\log _2}\left( {1 + {\gamma _{E - C}} } \right),
\end{equation}
where ${\gamma_{E - C}}$ is given in (\ref{eq:SNR_E-C}).

To maximize the throughput of this protocol, we formulate the following optimization problem:
\begin{equation}\label{}
\left({\rm{P3.1}}\right):\begin{array}{l}
\mathop {\max }\limits_{{P_A},{P_{R}},\tau_1,\tau_2} {{\mathcal T}_{E - C}}\\
\;\;\;\;{\rm{s}}{\rm{.t}}{\rm{.}}\;(\ref{eq:power_constraint_AP}),\;(\ref{eq:power_constraint_relay}),\;
\tau_1 + \tau_2 \le 1,\\
\;\;\;\;\;\;\;\;\;\;P_A,P_R,\tau_1,\tau_2 \ge 0.
\end{array}
\end{equation}

Unfortunately, it is easy to check that the problem (P3.1) is not a convex one. To tackle the non-convexity, we introduce two new variable $E_A = \tau_1 P_A$ and $E_R = \tau_1 P_R$. Based on this variable substitution, the throughput of the E-C protocol can be rewritten as
\begin{equation}\label{eq:throughput_E-C_rewritten}
{{\mathcal T}_{E - C}^\prime} =  \tau_2 {\log _2}\left( {1 + \frac{{\eta  \left( {E_A{h_{AS}} + E_R{h_{RS}}} \right){h_{SA}}}}{{\tau_2{N_0}}}} \right).
\end{equation}
Accordingly, the problem (P3.1) can be reformulated as
\begin{equation}\label{}
\left({\rm{P3.2}}\right):\begin{array}{l}
 \mathop {\max }\limits_{{E_A},{E_{R}},\tau_1,\tau_2} {{\mathcal T}_{E - C}^\prime} \\
\;\;\;\;{\rm{s}}{\rm{.t}}{\rm{.}}\;E_A \le \tau_1 P_A^{\max},\;E_A \le {P_A^{\rm{avg}}}, \\
\;\;\;\;\;\;\;\;\;\;E_R \le \tau_1 P_R^{\max},\;E_R \le {P_R^{\rm{avg}}}, \\
\;\;\;\;\;\;\;\;\;\;\tau_1 +\ \tau_2 \le 1, \\
\;\;\;\;\;\;\;\;\;\;E_A,E_R,\tau_1,\tau_2 \ge 0.
 \end{array}
\end{equation}

To solve the problem (P3.2), we first consider its simplified problem by removing the constraints
\begin{equation}\label{eq:inequality_condition_E_A_E_R}
E_A \le {P_A^{\rm{avg}}},\;E_R \le {P_R^{\rm{avg}}}.
\end{equation}
In this case, we have the following problem:
\begin{equation}\label{}
\left({\rm{P3.2.1}}\right):\begin{array}{l}
 \mathop {\max }\limits_{{E_A},{E_{R}},\tau_1,\tau_2} {{\mathcal T}_{E - C}^\prime} \\
\;\;\;\;{\rm{s}}{\rm{.t}}{\rm{.}}\;{E_A} \le \tau_1 P_A^{\max},\; {E_R} \le \tau_1 P_R^{\max},\\
\;\;\;\;\;\;\;\;\;\;\tau_1 +\ \tau_2 \le 1, \\
\;\;\;\;\;\;\;\;\;\;E_A,E_R,\tau_1,\tau_2 \ge 0.
 \end{array}
\end{equation}
It is straightforward to see that the throughput ${{\mathcal T}_{E - C}^\prime}$ in (\ref{eq:throughput_E-C_rewritten}) is monotonically increasing with $E_A$ and $E_R$ for given values of $\tau_1$ and $\tau_2$. Then, we can deduce that the optimal $E_A$ and $E_R$ should satisfy
\begin{equation}\label{eq:optimal_condition_E_A_E_R}
{E_A} = \tau_1 P_A^{\max},\;{E_R} = \tau_1 P_R^{\max}.
\end{equation}
Accordingly, we can further simplify the problem (P3.2.1) to the following one regarding time allocation only:
\begin{equation}\label{}
\left({\rm{P3.2.2}}\right):\begin{array}{l}
 \mathop {\max }\limits_{\tau_1,\tau_2} {{\mathcal T}_{E - C}^\prime} \\
\;{\rm{s}}{\rm{.t}}{\rm{.}}\;\tau_1 +\ \tau_2 \le 1,\;\tau_1,\tau_2 \ge 0.
 \end{array}
\end{equation}

The above problem (P3.2.2) can be regarded as a special case of the one addressed in \cite{Ju_TWC_2014}. Following the analyses in \cite{Ju_TWC_2014}, we can steadily obtain the optimal solution of the problem ({P3.2.1}) given by
\begin{subequations}\label{eq:optimal_solution_P3.2.1}
\begin{align}
&\tau_1^\bullet = \frac{{{z^\bullet} - 1}}{{A + {z^\bullet} - 1}},\;
\tau_2^\bullet = 1 - \tau_1^\bullet,\\
&E_A^\bullet = \tau_1^\bullet P_A^{\max},\;
E_R^\bullet = \tau_1^\bullet P_R^{\max},
\end{align}
\end{subequations}
where $z^\bullet$ is the unique solution of the equation
\begin{equation}\label{}
z\ln z - z + 1 = \frac{{\eta  \left( {P_A^{\max}{h_{AS}} + P_R^{\max}{h_{RS}}} \right){h_{SA}}}}{{{N_0}}}.
\end{equation}

Based on the above analyses, we can obtain the following proposition in terms of the optimal solution to the original problem (P3.2):
\begin{proposition}\label{prop:optimal_solution_3.1}
The optimal solution to the problem (P3.2), denoted by $\left(E_A^*,E_R^*, \tau_1^*,\tau_2^*\right)$, is given by
\begin{equation}\label{eq:optimal_solution_P3.2_E}
E_X^* = \left\{ \begin{array}{l}
\tau_1^\bullet P_X^{\max },\;{\rm{ if }}\;{\tau _1^\bullet} \le \mu, \\
{P_X^{\rm{avg}}},\;\;\;\;\;\;\;\;{\rm{if }}\;{\tau _1^\bullet} > \mu,  \\
 \end{array} \right.\;X \in \left\{ {A,R} \right\},
\end{equation}
\begin{equation}\label{eq:optimal_solution_P3.2_tau}
 \tau _1^* = \left\{ \begin{array}{l}
 {\tau _1^\bullet},\;{\rm{ if }}\;{\tau _1^\bullet} \le \mu,  \\
 \mu ,\;\;\;{\rm{ if }}\;{\tau _1^\bullet} > \mu,  \\
 \end{array} \right.,\; \tau _2^* = 1 - \tau _1^*,
\end{equation}
where $\mu$ is defined in (\ref{eq:def_mu}).
\end{proposition}
\begin{proof}
Firstly, it is easy to verify that if $\tau_1^\bullet \le \mu$, the optimal solution in (\ref{eq:optimal_solution_P3.2.1}) can also achieve the maximum of the problem (P3.2) without violating the conditions in (\ref{eq:inequality_condition_E_A_E_R}).

For the case when $\tau_1^\bullet > \mu$, however, the optimal solution in (\ref{eq:optimal_solution_P3.2.1}) violates the conditions in (\ref{eq:inequality_condition_E_A_E_R}). In this case, the optimal $E_A$ and $E_R$ should satisfy that $E_X^* = P_X^{\rm{avg}},\; X\in{A,R}$ regardless the value of $\tau_1$. Moreover, it can be shown that the condition $\tau_1^* + \tau_2^* =1$ should be met by the optimal $\tau_1^*$ as well as $\tau_2^*$, and the objective function of problem (P3.2) is monotonically increasing with $\tau_2$. Thus, the value of $\tau_1$ should be as small as possible. Thus, $\tau_1^* = \frac{E_X^*}{P_X^{\max}} = \mu$ and $\tau_1^* = 1 - \tau_2^*$. This completes the proof.
\end{proof}

Then, we can find the optimal values of $P_A$ and $R_R$ for the original problem (P3.1) by performing $P_X^* = E_X^*/\tau_1^*$.

\begin{remark}
It is interesting to notice that the optimal $P_X^* = P_X^{\max}$ for any value of~$\tau_1^*$. In other words, the AP and relay in the E-C protocol should always transmit with the peak power regardless the value of the optimal time allocation.
\end{remark}

\subsection{Throughput Maximization for D-C Protocol}
In D-C protocol, the relay power needs to be split into two fractions that are respectively used for DL energy transfer and UL information forwarding. Analogous to the previous subsection, we can formulate the following throughput maximization problem in terms of power allocation and time allocation for the D-C protocol:
\begin{equation}\label{eq:opt_D-C}
\left({\rm{P3.3}}\right):\begin{array}{l}
 \mathop {\max }\limits_{{P_A},{P_{R}^D},{P_{R}^U},\tau_1,\tau_2} {{\mathcal T}_{D - C}} \\
\;\;\;\;{\rm{s}}{\rm{.t}}{\rm{.}}\;(\ref{eq:power_constraint_AP}),\;(\ref{eq:power_constranit_Relay_D-C}),\; \tau_1 + \tau_2 \le 1,\\
\;\;\;\;\;\;\;\;\;\;P_A,P_R^D,P_R^U,\tau_1,\tau_2 \ge 0.
 \end{array}
\end{equation}
where ${{\mathcal T}_{D - C}}$ denotes the throughput of the D-C protocol given~by
\begin{equation}\label{}
{{\mathcal T}_{D - C}} = \frac{{\tau_2 }}{2}{\log _2}\left( {1 + {\gamma _{D - C}}} \right)
\end{equation}
with $\gamma _{D - C}$ defined in (\ref{eq:SNR_e2e_D-C}).

To proceed, we introduce three new variables defined as $E_A = \tau_1 P_A$, $E_R^D = \tau_1 P_R^D$ and $E_R^U = \frac {\tau_2}{2} P_R^U$. Furthermore, it can be shown that ${{\mathcal T}_{D - C}}$ increases with $\tau_1$ for a fixed $\tau_2$ and increases with $\tau_2$ with a fixed $\tau_1$. This means that the optimal values of $\tau_1$ and $\tau_2$ should satisfy $\tau_1 + \tau_2 =1$. Then, we can remove one of the variables and reformulate the problem (P3.3) as
\begin{equation}\label{}
\left({\rm{P3.4}}\right):\begin{array}{l}
 \mathop {\max }\limits_{{E_A},{E_{R}^D},{E_{R}^U},\tau_1} {{\mathcal T}_{D - C}^\prime} \\
\;\;\;\;{\rm{s}}{\rm{.t}}{\rm{.}}\;E_A \le \tau_1 P_A^{\max},\;E_A \le {P_A^{\rm{avg}}}, \\
\;\;\;\;\;\;\;\;\;\;E_R^D \le \tau_1 P_R^{\max},\;E_R^U \le \frac{1-\tau_1}{2} P_R^{\max}, \\
\;\;\;\;\;\;\;\;\;\;E_R^D + E_R^U \le {P_R^{\rm{avg}}}, \\
\;\;\;\;\;\;\;\;\;\;E_A,E_R^D,E_R^U,\tau_1 \ge 0,
 \end{array}
\end{equation}
where
\begin{equation}\label{}
{{\mathcal T}_{D - C}^\prime} = \frac{{1-\tau_1 }}{2}{\log _2}\left( {1 + {\gamma _{D - C}^\prime}} \right)
\end{equation}
with
\begin{equation}\label{}
\begin{split}
\gamma _{D - C}^\prime  = &\frac{{2\eta \left( {{E_A}{h_{AS}} + E_R^D{h_{RS}}} \right){h_{SA}}}}{{\left( {1 - {\tau _1}} \right){N_0}}} +\\
&\frac{{\frac{{2\eta \left( {{E_A}{h_{AS}} + E_R^D{h_{RS}}} \right){h_{SR}}}}{{\left( {1 - {\tau _1}} \right){N_0}}}\frac{{2E_R^U{h_{RA}}}}{{\left( {1 - {\tau _1}} \right){N_0}}}}}{{\frac{{2\eta \left( {{E_A}{h_{AS}} + E_R^D{h_{RS}}} \right){h_{SR}}}}{{\left( {1 - {\tau _1}} \right){N_0}}} + \frac{{2E_R^U{h_{RA}}}}{{\left( {1 - {\tau _1}} \right){N_0}}} + 1}}.
\end{split}
\end{equation}

However, the simplified problem (P3.4) is still hard to address due to the complexity of the objective function. To resolve this, we adopt the following method: we first solve the problem (P3.4) for a given value of $\tau_1$ and then find the optimal $\tau_1$ via numerical method (e.g., one-dimensional exhaustive search). After a careful observation on its structure, the problem (P3.4) can be simplified to the following three problems based on the given value of $\tau_1$:

\textbf{(1) When} $0 \le \tau_1\le 2\mu -1$: Note that this case happens only if $\mu\ge 0.5$. For any $\tau_1 \in \left[0,2\mu-1\right]$, it is evident that the average power constraints $E_A \le {P_A^{\rm{avg}}}$ and $E_R^D + E_R^U \le {P_R^{\rm{avg}}}$ can be removed. Moreover, ${{\mathcal T}_{D - C}^\prime}$ is shown to be monotonically increasing with $E_A$, $E_R^D$ and $E_R^U$, respectively. Thus, the optimal values for $E_A$, $E_R^D$ and $E_R^U$ are given by
\begin{equation}\label{eq:optimal_E's_1}
{E_A^\circ} = {\tau _1}P_A^{\max } ,\; E_R^{D ,\circ} = {\tau _1}P_R^{\max }, \; E_R^{U,\circ} = \frac{{1 - {\tau _1}}}{2}P_R^{\max } .
\end{equation}

\textbf{(2) When} $2\mu -1 < \tau_1 \le \mu$: The constraint $E_A \le {P_A^{\rm{avg}}}$ can still be ignored and the optimal value of $E_A$ is still given by ${E_A^\circ} = {\tau _1}P_A^{\max } $. But, the constraint $E_R^D + E_R^U \le {P_R^{\rm{avg}}}$ should be considered and updated as $E_R^D + E_R^U = {P_R^{\rm{avg}}}$. We define an auxiliary variable $t = E_R^U /E_R^D$ to facilitate the problem solving. Then, we can reformulate the problem (P3.4) with a given $\tau_1$ as
\begin{equation}\label{}
\left({\rm{P3.4.1}}\right):\begin{array}{l}
 \mathop {\max }\limits_{t} {{\gamma}_{D - C}^{\prime\prime}} \\
\;{\rm{s}}{\rm{.t}}{\rm{.}}\;t_L \le t \le t_U,
 \end{array}
\end{equation}
where
\begin{equation}\label{}
\gamma _{D - C}^{\prime \prime } = a + \frac{b}{{t + 1}} + \frac{{\left( {c + \frac{d}{{t + 1}}} \right)\frac{{et}}{{t + 1}}}}{{c + \frac{d}{{t + 1}} + \frac{{et}}{{t + 1}} + 1}},
\end{equation}
\begin{equation}\label{}
{t_L} = \left({{\mu  - {\tau _1}}}\right)/{{{\tau _1}}},
\end{equation}
\begin{equation}\label{}
{t_U} = \left\{ \begin{array}{l}
 \left( {1 - {\tau _1}} \right)/\left( {2\mu  - 1 + {\tau _1}} \right),\;{\rm{if}}\;{\tau _1} > 1 - 2\mu,  \\
 \infty ,\;{\rm{otherwise}}, \\
 \end{array} \right.
\end{equation}
with $ a = \frac{{2\eta E_A^\circ {h_{AS}}{h_{SA}}}}{{\left( {1 - {\tau _1}} \right){N_0}}}$, $b = \frac{{2\eta {P_R^{\rm{avg}}}{h_{RS}}{h_{SA}}}}{{\left( {1 - {\tau _1}} \right){N_0}}}$, $c = \frac{{2\eta E_A^\circ {h_{AS}}{h_{SR}}}}{{\left( {1 - {\tau _1}} \right){N_0}}}$, $d = \frac{{2\eta {P_R^{\rm{avg}}}{h_{RS}}{h_{SR}}}}{{\left( {1 - {\tau _1}} \right){N_0}}}$, and $e = \frac{{2{P_R^{\rm{avg}}}{h_{RA}}}}{{\left( {1 - {\tau _1}} \right){N_0}}}$, which are defined for the notation simplicity.

After some algebraic manipulations, we can obtain the following lemma regarding the optimal solution to the problem (P3.4.1):
\begin{lemma}
The optimal solution to the problem (P3.4.1) can be expressed as
\begin{equation}\label{eq:optimal_solution_P3.4.1}
{t^*} = \arg \max _{{t \in \left\{ {{t_{L}},{t_1},{t_2},{t_U}} \right\}}}\gamma _{D - C}^{\prime \prime },
\end{equation}
where
\begin{subequations}\label{}
\begin{align}
&{t_1} = \left\{ \begin{array}{l}
 {t_1^\prime},\;{\rm{if}}\;\bigtriangleup \ge 0\;{\rm{and}}\;{t_L} \le {t_1^\prime} \le {t_U} ,\\
 \emptyset,\; \;{\rm{otherwise}}, \\
 \end{array} \right.\\
&{t_2} = \left\{ \begin{array}{l}
 {t_2^\prime},\;{\rm{if}}\;\bigtriangleup \ge 0\;{\rm{and}}\;{t_L} \le {t_2^\prime} \le {t_U} ,\\
 \emptyset,\; \;{\rm{otherwise}}, \\
 \end{array} \right.,
\end{align}
in which, $\bigtriangleup = B^2 -4AC$, $t_1^\prime  = \frac{{ - b + \sqrt \Delta  }}{{2A}}$, $t_2^\prime  = \frac{{ - b - \sqrt \Delta  }}{{2A}}$. Here, $A = ce - 2bc - 2be - b - de - bc^2 - be^2 + c^2e - de^2 - 2bce$, $B = 2ce - 4bc - 2bd - 2be - 2b - 2bc^2 + 2c^2e - 2bcd - 2bce - 2bde + 2cde$, $C = ce - 2bc - 2bd - b + de - bc^2 - bd^2 + c^2e + d^2e - 2bcd + 2cde$.
\end{subequations}
\end{lemma}
\begin{proof}
We calculate the first-order derivative of $\gamma _{D - C}^{\prime \prime }$ with respect to $t$ and obtain that
\begin{equation}\label{}
\partial \gamma _{D - C}^{\prime \prime }/\partial t \propto A{t^2} + Bt + C,
\end{equation}
which means that $\gamma _{D - C}^{\prime \prime }$ has up to two extreme points in terms of $t$ without considering the constraint. Thus, the maximizer of problem (P3.4.1) can be easily obtained through evaluating the values of $\gamma _{D - C}^{\prime \prime }$ at feasible extreme points and two limits. Mathematically, we have (\ref{eq:optimal_solution_P3.4.1}), which completes the proof.
\end{proof}

Then, the optimal values for $E_A$, $E_R^D$ and $E_R^U$ are accordingly given by
\begin{equation}\label{eq:optimal_E's_2}
{E_A^\circ} = {\tau _1}P_A^{\max } ,\; E_R^{D ,\circ} = \frac{P_R^{\rm{avg} }}{t^* + 1}, \; E_R^{U,\circ} = \frac{t^* P_R^{\rm{avg} }}{t^* + 1} .
\end{equation}
\textbf{(3) When} $\tau_1 > \mu$: In this scenario, the two average power constraints for the AP and relay are both active and updated as $E_A = {P_A^{\rm{avg}}}$ and $E_R^D + E_R^U = {P_R^{\rm{avg}}}$. However, the two constraints $E_A \le \tau_1 P_A^{\max}$ and $E_R \le \tau_1 P_R^{\max}$ can be ignored. Following the similar analysis as in the proof of Proposition \ref{prop:optimal_solution_3.1}, we can deduce that when the allocation parameter $t$ of the relay is given, the maximum energy harvested by the source is fixed for any $\tau_1$ that is no less than $\mu$. In this case, the time allocated for energy transfer should be as small as possible. Intuitively, we have the following lemma:
\begin{lemma}
For any $\tau_1 \in \left(\mu,1\right]$, the corresponding maximal throughput is less than that of the case when $\tau_1 = \mu$.
\end{lemma}
Note that the above lemma reveals that the interval $\left(\mu,1\right]$ is not needed to consider when we calculate the optimal value of $\tau_1$.

By combining the three cases analyzed above, we can obtain the optimal solution to the original problem (P3.4) given in the following proposition:
\begin{proposition}
The optimal value for $\tau_1$ of the problem (P3.4) can be expressed as
\begin{equation}\label{}
\tau _1^* = \arg {\max _{{\tau _1} \in \left[ {0,\mu } \right]}}{\cal T}_{D - C}^\prime \left( {{E_A^\circ},E_R^{D ,\circ},E_R^{U ,\circ}} \right),
\end{equation}
where ${E_A^\circ}$, $E_R^{D ,\circ}$, and $E_R^{U ,\circ}$ are given in (\ref{eq:optimal_E's_1}) or (\ref{eq:optimal_E's_2}) based on the value of $\tau_1$. Accordingly, the optimal values for other parameters can be calculated via $P_A^* = \frac{{E_A^\circ \left( {\tau _1^*} \right)}}{{{\tau _1}}}$, $P_R^{D,*} = \frac{{E_R^{D,^\circ }\left( {\tau _1^*} \right)}}{{{\tau _1}}}$, $\tau _2^* = 1 - \tau _1^*$, and $P_R^{U,*} = \frac{{2E_R^{U,^\circ }\left( {\tau _1^*} \right)}}{{{\tau _2}}}$.
\end{proposition}

\begin{remark}
It is worth noting that although the closed-form optimal solution to the problem (P3.4) with five variables is not given, this problem can be efficiently solved via one-dimensional exhaustive search in the proposed method. Moreover, our analyses reduce the interval of the exhaustive search.
\end{remark}

\section{Numerical Results}
In this section, we present some numerical results to illustrate and compare the performance of the proposed protocols. To obtain meaningful results, we restrict our attention to a linear topology. Specifically, the relay is located on a straight line between the AP and source, i.e, $d_{AR}=  d_{AS} - d_{SR}$ with $d_{XY}$ denoting the distance between nodes $X$ and $Y$. The channel short-term fading is assumed to be Rayleigh distributed. To capture the effect of path-loss on the network performance, we use the channel model that ${\mathbb{E}}\left\{h_{XY}\right\} = 10^{-3}\left({d_{XY}}\right)^{ - \alpha }$, where $\alpha \in [2,5]$ is the path-loss factor \cite{Chen_SPL_2010}. Note that a 30dB average signal power attenuation is assumed at a reference distance of $1$m in the above channel model \cite{Ju_TWC_2014}. In all following simulations, we set equal average transmit power for the AP and relay, the distance between the AP and source $d_{AS} = 10$m, the path-loss exponent $\alpha$ = 2, the noise power $N_0 = -80$dBm, and the energy harvesting efficiency $\eta = 0.5$. Moreover, each curve for the average throughput is obtained by averaging over $5000$ randomly generated channel realizations.
\begin{figure}
\centering \scalebox{0.45}{\includegraphics{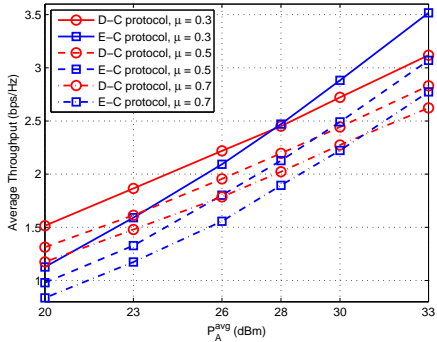}}
\caption{The average throughput of the proposed protocols versus the average transmit power of the AP (i.e., $P_A^{\rm{avg}}$), where $d_{SR} = 5$m and $P_R^{\rm{avg}} = P_A^{\rm{avg}}$. \label{fig:Throughput_SNR}}
\end{figure}

Fig. \ref{fig:Throughput_SNR} plots the average throughput curves of the E-C and D-C protocols versus the average transmit power of the AP with different values of $\mu$, where the relay is located in the middle of the AP and source. We can see that the performance of both protocols increases monotonically with the average transmit power of the AP for any value of $\mu$. For both E-C and D-C protocols, we can observe that the average throughput decreases as the parameter $\mu$ increases. This is because that for a given average transmit power, the peak transmit power decreases when $\mu$ increases, which reduces the feasible sets of the transmit powers and thus degrades the throughput performance. It can also be observed from Fig. \ref{fig:Throughput_SNR} that the D-C protocol is superior to the E-C protocol when the average transmit power is relatively small to medium. But this observation is reversed when the average transmit power is high enough. This is understandable since the throughput is highly affected by the information transmission time at high SNR and the time utilization of E-C protocol is better than that of the D-C protocol. Furthermore, higher average transmit power is needed for the E-C protocol to outperform the D-C protocol when the value of $\mu$ grows.

\begin{figure}
\centering \scalebox{0.45}{\includegraphics{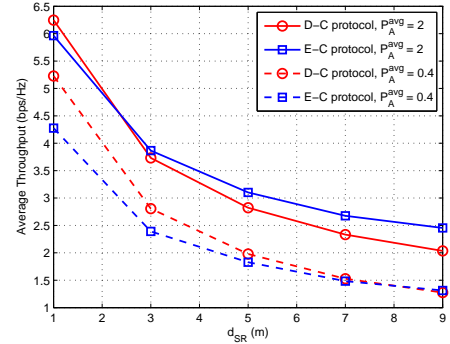}}
\caption{The average throughput of the proposed protocols versus $d_{SR}$, where $\mu = 0.5$ and $P_R^{\rm{avg}} = P_A^{\rm{avg}}$. \label{fig:Throughput_d_SR}}
\end{figure}
Fig. \ref{fig:Throughput_d_SR} depicts the impact of the relay position on the average throughput of the proposed protocols, in which the throughput curves are plotted versus the distance between the source and relay (i.e., $d_{SR}$) with two different values of $P_A^{\rm{avg}}$. From Fig. \ref{fig:Throughput_d_SR}, we can observe that the average throughput of both protocols decreases smoothly with the increasing of $d_{SR}$. This observation indicates that the hybrid relay should be deployed nearer to the source to obtain better throughput. Besides, it is observed from Fig. \ref{fig:Throughput_d_SR} that when the average transmit power is equal to $2$Watt (i.e., at high SNR), the E-C protocol is superior to the D-C protocol unless the relay is very close to the source. In contrast, in lower SNR regime (i.e., $P_A^{\rm{avg}}$ is $0.4$Watt), the D-C protocol outperforms the E-C protocol until the relay is very far away from the source.

\section{Conclusions}
In this paper, two cooperative protocols, energy cooperation (E-C) and dual cooperation (D-C), were developed for a wireless-powered cooperative communication network consisting of a hybrid AP, a hybrid relay and an energy harvesting source. The throughput maximization problems in terms of joint power and time allocation were formulated and resolved for the proposed protocols. Numerical results showed that the (E-C) protocol achieves better throughput at high signal-to-noise ratios (SNRs), especially when the distance between the source and relay is large. In contrast, when the SNR is not high and the relay is relatively close to the source, the D-C protocol is superior to the E-C protocol.


\ifCLASSOPTIONcaptionsoff
  \newpage
\fi

\bibliographystyle{IEEEtran}
\bibliography{References}

%

\end{document}